\documentclass[a4paper,11pt]{article}

\newtheorem{theorem}{Theorem}[section]
\newtheorem{lemma}[theorem]{Lemma}
\newtheorem{proposition}[theorem]{Proposition}

\newtheorem{definition}[theorem]{Definition}

\newenvironment{proof}[1][Proof]{\begin{trivlist}
\item[\hskip \labelsep {\bfseries #1}]}{\end{trivlist}}

\newenvironment{remark}[1][Remark]{\begin{trivlist}
\item[\hskip \labelsep {\bfseries #1}]}{\end{trivlist}}

\newcommand{\qed}{\nobreak \ifvmode \relax \else
      \ifdim\lastskip<1.5em \hskip-\lastskip
      \hskip1.5em plus0em minus0.5em \fi \nobreak
      \vrule height0.75em width0.5em depth0.25em\fi}

\newcommand{\Dleft}{[\hspace{-1.5pt}[}
\newcommand{\Dright}{]\hspace{-1.5pt}]}
\newcommand{\SN}[1]{\Dleft #1 \Dright}

\newcommand{\tombstone}{\nobreak \ifvmode \relax \else
      \ifdim\lastskip<1.5em \hskip-\lastskip
      \hskip0em plus0em minus0.5em \fi \nobreak
      \vrule height0.75em width0.5em depth0.25em\fi}

\newcommand{\proofend}{\flushright $\qed$}

\usepackage{amssymb}
\usepackage{hyperref}
\usepackage{pstricks}
\usepackage{amsmath}
\usepackage{dsfont}
\usepackage{fancyhdr}

\numberwithin{equation}{section}

\DeclareMathOperator{\Can}{Can}
\DeclareMathOperator{\Vect}{Vect}
\DeclareMathOperator{\Diff}{Diff}

\begin{document}
\bibliographystyle{plain}

\author{Andrew James Bruce \footnote{e-mail: \texttt{andrewjames.bruce@physics.org}}\\ \small \emph{School of Mathematics}\\ \small \emph{The University of Manchester}\\
\small \emph{Manchester}\\\small \emph{M13 9PL}}
\date{\today}
\title{On higher Poisson and Koszul--Schouten brackets}
\maketitle

\begin{abstract}
\noindent In this note we show how to construct a  homotopy BV-algebra on the algebra of differential forms over a higher Poisson manifold. The Lie derivative  along the higher Poisson structure provides the generating operator.
\end{abstract}

\noindent\begin{small} \textbf{Keywords:} Strong homotopy Lie algebra, derived bracket, higher Poisson bracket, higher Koszul--Schouten bracket, Batalin--Vilkoviski algebra, odd symplectic geometry.\end{small}\\

\section{Introduction and Background}\label{section1}

\noindent  Recall that a higher Poisson manifold is the pair $(M, P)$, where $M$ is a supermanifold and  the higher Poisson structure $P \in \mathfrak{X}^{*}(M) (= C^{\infty}(\Pi T^{*}M))$ is an even parity (but otherwise arbitrary) multivector field that satisfies the ``classical master equation"; $\SN{P,P}=0$. Here the bracket is the canonical Schouten--Nijenhuis bracket on $\Pi T^{*}M$. The original formulation of higher Poisson manifolds was laid down by Voronov \cite{voronov-2004,voronov-2005}. A similar notion was put forward by de Azc$\acute{\textnormal{a}}$rraga et.al \cite{deAzcarraga:1996zk,deAzcarraga:1996jk}. \\

\noindent Associated with the higher Poisson structure is a homotopy Poisson algebra on $C^{\infty}(M)$. That is an $L_{\infty}$-algebra in the sense of Lada \& Stasheff  \cite{Lada:1992wc} (suitably ``superised") such that the series of brackets are multi-derivations over the (supercommutative) product of functions.  The series of higher Poisson brackets are given by;

 \begin{equation}
 \{f_{1}, f_{2}, \cdots, f_{r} \}_{P} = \left.\Dleft \cdots\Dleft \Dleft P,f_{1}\Dright,f_{2}\Dright \cdots, f_{r}\Dright \right|_{M},
\end{equation}

\noindent where $f_{I} \in C^{\infty}(M)$. The ``classical master equation" is directly equivalent to the higher Poisson brackets satisfying the so-called Jacobiators. This is fundamental as without this condition the series of higher brackets would not form an $L_{\infty}$-algebra.\\

 \noindent In this note we show how to construct a series of higher brackets on the algebra of differential forms\footnote{we define the algebra of differential forms over a supermanifold to be $\Omega^{*}(M) := C^{\infty}(\Pi TM)$. As we will not delve into the theory of integration over supermanifolds, such a definition is perfectly adequate for our purposes. } on a higher Poisson manifold. We do this by following Koszul \cite{Koszul;1985} (also see \cite{Akman:1995tm,Bering:1996kw,voronov-2004}), as \emph{higher antibrackets} generated by the Lie derivative along the higher Poisson structure. The $n+1$-th bracket in essence measures the failure of $n$-th  bracket to be a multi-derivation. As the Poisson structure is an even multivector field, the Lie derivative is an odd operator acting on differential forms and thus the series of brackets are odd. The ``classical master equation" guarantees that the Lie derivative along the higher Poisson structure is nilpotent. Such a series of brackets forms an $L_{\infty}$-algebra commonly referred to as a homotopy BV-algebra. \\

\noindent We will refer to this series of brackets as higher Koszul--Schouten brackets.  These brackets are considered as  the natural analogues of the classical Koszul--Schouten bracket \cite{Kosmann-Schwarzbach:1995,Kosmann-Schwarzbach:1996,Koszul;1985}. We show that the exterior derivative acts as a differential over the higher Koszul--Schouten brackets, just as it does in the classical case. Thus we have what we will call a differential homotopy BV-algebra.\\

\noindent It must be noted that Khudaverdian \& Voronov \cite{khudaverdian-2008} constructed a homotopy Schouten algebra on the algebra of differential forms over a higher Poisson manifold. That is they have an odd analogue of a homotopy Poisson algebra constructed from the higher Poisson structure. Importantly, the higher Schouten brackets are multi-derivations over the algebra of differential forms. This is in stark contrast to the higher Koszul--Schouten brackets.\\

\noindent Thus, there is at least two ways to equip the algebra of differential forms over a higher Poisson manifold with the structure of an $L_{\infty}$-algebra. However, these two constructions are not completely independent. Due to ideas present in \cite{voronov-2004}, the higher Schouten brackets can be thought of as the ``classical limit" of the Koszul--Schouten brackets. It is not surprising that some relation exists between the two classes of higher bracket as both are build from the initial geometric data of a higher Poisson structure. \\

\noindent The distinction between Koszul--Schouten and Schouten brackets in the case  of classical (binary) Poisson manifolds is non-existent. Simply put, there is no tri-bracket or higher and the distinction is not present.\\

\noindent Although we will not present details here, it is clear that the constructions presented in this note carry over to higher Poisson structures on Lie algebroids.\\

\noindent However, it is clear that a direct analogue for higher Schouten manifolds does not exist. Associated to such a structure is a homotopy Poisson algebra on differential forms over the supermanifold. One can construct the higher Poisson structure  using arguments almost identical to that found in \cite{khudaverdian-2008}.  As such brackets are not all odd, no operator can generate something resembling a homotopy Poisson algebra in some ``classical limit". \\

\noindent As is customary in \emph{super-mathematics}, the prefix ``super" will generally be omitted. For example, by manifold we explicitly mean supermanifold.  We will denote the parity of an object $A$, by \emph{tilde}; $\widetilde{A} \in \mathds{Z}_{2}$. By \emph{even} or \emph{odd} we will be explicitly referring to parity. By a \emph{Q-manifold} will mean the pair $(M,Q)$ where $M$ is a supermanifold and $Q \in \Vect(M)$ is an odd vector field known as a \emph{homological vector field} satisfying $[Q,Q]=0$. \\

\noindent The notion of homotopy Poisson/Schouten/BV algebra as used here is much more restrictive that that found elsewhere in the literature \cite{galvezcarrillo-2009,Ginzburg-1994}. We make no use of the theory of operads. However, the notion used throughout this work seems very well suited to geometric considerations and suits the purposes explored here. \\

\newpage

\section{Higher Koszul--Schouten brackets}\label{section2}

\noindent Differential forms on a manifold $M$,  are understood as functions on antitangent bundle  $\Pi TM$. In natural local coordinates $\{ x^{A}, dx^{A}\}$ on $\Pi TM$ an arbitrary  differential form is locally given by $\alpha(x,dx) \in \Omega^{*}(M) = C^{\infty}(\Pi TM)$. The parities of the local coordinates  are $\widetilde{x}^{A} = \widetilde{A}$ and $d\widetilde{x}^{A} = (\widetilde{A}+1)$.  Under transformations induced by coordinate changes on the base the fibre coordinates transform as $d\overline{x}^{A} = dx^{B} \left( \frac{\partial \overline{x}^{A}}{\partial x^{B}}\right)$.\\

\noindent Similarly, multivector fields on $M$, are identified with functions on the anticotangent bundle  $\Pi T^{*}M$. We pick  natural local coordinates $\{x^{A}, x^{*}_{A}\}$, such that  $\widetilde{x}^{*}_{A} = (\widetilde{A}+1)$. Multivector fields  on $M$ are locally described  by $X(x,x^{*}) \in  \mathfrak{X}^{*}(M) = C^{\infty}(\Pi T^{*}M)$. The fibre coordinates transform as $\overline{x}^{*}_{A} = \left( \frac{\partial x^{B}}{\partial \overline{x}^{A}} \right)x^{*}_{B}$ under transformations induced by coordinate changes on the base. A filtration can be defined on multivectors (and differential forms) that are polynomial in fibre coordinates. We will refer to this as multivector degree.  Where no confusion will arise we will simply use degree for multivector degree.   We will call a multivector field an $r$-vector if it is a monomial of  degree $r$ in the fibre coordinates.\\

\noindent \textbf{Warning:} The parity of an object is generally independent of the degree. By odd and even we will be explicitly referring to parity and not weight.\\

\noindent Importantly, the anticotangent bundle comes equipped with a canonical odd symplectic structure. That is we have an odd Poisson bracket on $\mathfrak{X}^{*}(M)$. This bracket is the \textbf{Schouten--Nijenhuist bracket}.\\

\noindent From an arbitary multivector field $X \in \mathfrak{X}^{*}(M)$ we can associate a differential operator\footnote{Generically it will in fact be a pseudo-differential operator if we do not restrict attention to multivectors that are polynomial in  $x^{*}$.} acting on differential forms on $M$ known as the \textbf{Lie derivative};

\begin{equation}
X \rightsquigarrow L_{X} =     [d, i_{X}],
\end{equation}

\noindent where $d = dx^{A}\frac{\partial}{\partial x^{A}}$ is the \textbf{exterior derivative} and $i_{X} = (-1)^{\tilde{X}}X(x, \partial/\partial dx)$ is the \textbf{interior product}. As a differential operator, the Lie derivative is of order equal to that of the degree of the multivector field (assuming it is defined).\\

\noindent The  Lie derivative  is clearly not a derivation on the algebra of differential forms, apart form the isolated case of vector fields. Instead of the  derivation property we have;

\begin{equation}
 L_{\SN{X,Y}} = [L_{X}, L_{Y}].
\end{equation}

\noindent Thus we see directly that for a  higher Poisson structure $P$, the Lie derivative is nilpotent;
\begin{equation}
(L_{P})^{2}=  \frac{1}{2}[L_{P} ,L_{P}] = \frac{1}{2}L_{\SN{P,P}} =0.
\end{equation}

\noindent A further property needed later on is that the exterior derivative is natural with respect to the Lie derivative, that is;

\begin{equation}
[d, L_{X}] = 0.
\end{equation}

\begin{remark}
The Lie derivative of a differential form along a multivector field should be understood as the infinitesimal action of $\Can(\Pi T^{*}M)$ on differential forms on $M$ (functions on $\Pi TM$). See \cite{khudaverdian-2006,Khudaverdian:2000zt,Khudaverdian:1999mz,Schwarz:1992nx} for more details.
\end{remark}

\begin{definition}
Let $(M,P)$ be a higher Poisson manifold. The \textbf{higher Koszul--Schouten brackets}  on $\Omega^{*}(M)$ are the higher derived brackets generated by the operator $L_{P} = [d,i_{P}]$,
\begin{equation}
[\alpha_{1}, \alpha_{2}, \cdots , \alpha_{r}]_{P}  := [\cdots[[L_{P},\alpha_{1}],\alpha_{2}]\cdots, \alpha_{r}]\mathds{1},
\end{equation}
with $\alpha_{i} \in \Omega^{*}(M)$.
\end{definition}

\noindent Here $\mathds{1}$ is the constant value $1$ zero form. That is the identity element in $\Omega^{*}(M)$. As the Lie derivative along the higher Poisson structure  is odd and  nilpotent the series on higher Koszul--Schouten brackets form a homotopy BV-algebra.\\

\noindent We know the Lie derivative $L_{P}$ is a differential operator of order equal to the degree (should it be defined) of $P$. The order of $L_{P}$ is at most $k$ if  $\Phi_{P}^{r}$ vanishes (for all $\alpha_{i}$) for $r \geq k+1$. Thus for any multivector of (finite) degree $k$, the above series of brackets terminates after  the $k$ place.  The top non-zero bracket is a multi-derivation. The $(r+1)$ bracket is  the obstruction to the Leibnitz rule for the $r$ bracket. However, it should be noted that we do not require polynomial multivector fields in order to define the higher Koszul--Schouten brackets.\\

\begin{proposition}
The exterior derivative satisfies a derivation rule over the Koszul brackets. Specifically,
\begin{equation}
d [\alpha_{1}, \cdots , \alpha_{r}]_{P} + \sum_{i=1}^{r}(-1)^{\epsilon_{i}}[\alpha_{1}, \cdots,d\alpha_{i}, \cdots \alpha_{r} ]_{P} =0,
\end{equation}
where
\begin{equation}
\epsilon_{j} = \left\{\begin{array}{cc}
0 & j = 1\\
\sum^{j-1}_{k=1} \widetilde{\alpha}_{k}&  j >1.
\end{array}\right.
\end{equation}
\end{proposition}

\begin{proof}
The above proposition follows directly from the naturality of the exterior derivative with respect to the generalised Lie derivative; $[d,L_{P}]=0$.
\proofend
\end{proof}

\noindent We take the attitude that the higher Koszul--Schouten brackets give a clear geometric example of ``higher antibrackets", \cite{Akman:1995tm, Alfaro:1995vw,Bering:2006eb,Bering:1996kw}. Slightly more than this, we have a ``differential homotopy BV-algebra" with the differential being provided by the exterior derivative.\\

\section{Relations between higher Poisson and higher Koszul--Schouten brackets}\label{section3}
\noindent In section we will generalise the known relations between the classical (binary) Poisson and Koszul--Schouten brackets. In order to achieve this we need to restrict our attention to finite degree Poisson structures. This is clear as although the higher Poisson brackets rely only on the Taylor expansion of the higher Poisson structure near $M \subset \Pi T^{*}M$ the higher Koszul--Schouten brackets require the full Poisson structure. First we need a lemma.\\

\begin{lemma}\label{lemma poisson bracket}
Let $P\in \mathfrak{X}^{*}(M)$ be a higher Poisson structure of degree $k$. Then the $r$-th higher Poisson bracket is given by
\begin{equation}
\{f_{1}, f_{2}, \cdots f_{r} \}_{P} = - L_{\stackrel{r}{P}}(f_{1} df_{2} \dots df_{r})=  i_{\stackrel{r}{P}}(df_{1}df_{2} \cdots  df_{r}).
\end{equation}
\end{lemma}

\begin{proof}
The above can be proved via direct computation.  \proofend
\end{proof}

\begin{proposition}\label{relation poisson koszul--schouten}
Let $P\in \mathfrak{X}^{*}(M)$ be a higher Poisson structure of degree $k$. Then the higher Koszul--Schouten brackets satisfy the following;
\begin{enumerate}
    \item $[f_{1}, df_{2} \cdots , df_{r}]_{P} =  -\{f_{1}, f_{2}, \cdots , f_{r} \}_{P}$,
    \item $[f_{1}, f_{2}, \cdots , f_{r}]_{P}= 0$ \hspace{20pt} for $r>1$,
    \item $[\emptyset]_{P} =  d\{ \emptyset\}_{P}$,
    \item $[df_{1}, df_{2}, \cdots , df_{r}]_{P}= d\{f_{1}, f_{2}, \cdots, f_{r}\}_{P}$,
\end{enumerate}
with $f_{I} \in C^{\infty}(M)$.
\end{proposition}

\begin{proof}
By counting the (form) degree of $[f_{1}, df_{2} \cdots , df_{r}]_{P}$ it is clear that the only contribution is from the $s \leq r$ components of $P$. Furthermore, we know that the only contribution to the $r$-bracket from a degree $r$ multivector is the top component. By expanding out the definition of the higher Koszul--Schouten brackets we see that we have the simplification

\begin{equation}
[f_{1}, df_{2} \cdots , df_{r}]_{P}= L_{\stackrel{r}{P}}(f_{1} df_{2} \cdots df_{r}),
\end{equation}

\noindent and via Lemma (\ref{lemma poisson bracket})  we obtain \emph{1}. As the higher Koszul--Schouten bracket carries form degree $(1-r)$ it is clear that for $r>1$ \emph{2}. holds. (The $r=1$ case is contained within \emph{1}.). For $r=0$ we have \emph{3}. which can easily be seen using local expressions. \emph{4}. follows directly from \emph{1}. and the derivation property of the exterior derivative.  \proofend
\end{proof}

\noindent The statements \emph{3}. and \emph{4}. in Proposition (\ref{relation poisson koszul--schouten}) show that the exterior derivative provides an morphism as $L_{\infty}$-algebras between the higher Poisson and higher Koszul--Schouten brackets.

\begin{remark}
The above relations between the higher Poisson and higher Koszul--Schouten brackets reduce to the case of classical Poisson structures for $2$-Poisson structures, up to conventions.
\end{remark}

\noindent If we have a classical Poisson manifold, then it is well known that the Poisson anchor provides a homeomorphism between the (binary) Koszul--Schouten bracket and the Schouten--Nijenhuis  bracket as

\begin{equation}
\phi^{*}_{P} [\alpha, \beta]_{P} = \SN{\phi_{P}^{*}\alpha, \phi_{P}^{*}\beta }.
\end{equation}

\noindent However, it is not at all obvious how this relation generalises to higher Poisson manifolds. In particular, the higher Koszul--Schouten brackets are a series of higher brackets as where there is only one (binary) Schouten--Nijenhuis. Thus the higher Poisson anchor cannot be a simple map between the brackets. It is expected that  this could be formulated in terms of an $L_{\infty}$-algebra morphism.

\section{From higher Koszul--Schouten brackets to higher Schouten brackets}\label{section4}

\noindent Khudaverdian \& Voronov  \cite{khudaverdian-2008} define a higher Schouten algebra on $\Omega^{*}(M)$. At first their constructions seem unrelated to constructions presented here.  The Koszul--Schouten brackets defined in this work provide the algebra of differential forms with the structure on a homotopy BV algebra. That is the  $n+1$ Koszul--Schouten bracket can be defined recursively as the failure of $n$ Koszul--Schouten bracket to be a multi-derivation. In  the case of higher  Schouten brackets one loses this recursive definition and replaces it with a strict Leibnitz rule. That is the higher Schouten brackets form a higher analouge of an \emph{odd Poisson algebra}. We shall see that the higher Schouten brackets of Khudaverdian \& Voronov are the \emph{classical limit} \cite{voronov-2004} of the higher  Koszul--Schouten brackets. To show this we need to ``deform" the higher Poisson structure as

\begin{equation}
P \rightsquigarrow P[\hbar] = \sum_{i=0}^{\infty} \frac{(\hbar)^{i}}{i!}P^{A_{1} \cdots A_{i}}(x) x^{*}_{A_{i}} \cdots x^{*}_{A_{1}},
\end{equation}

\noindent where $\hbar$ is an \emph{even formal deformation parameter}. In essence it counts the degree of the components of the multivector field. As usual, the infinite sum is understood formally.\\

\begin{definition}
The higher Schouten brackets associated with the higher Poisson structure \newline $P \in \mathfrak{X}^{*}(M)$ are defined as
\begin{equation}
(\alpha_{1},\alpha_{2}, \cdots , \alpha_{r})_{P} = \lim_{\hbar \rightarrow 0} \left(\frac{1}{\hbar}\right)^{r} [\alpha_{1}, \alpha_{2}, \cdots, \alpha_{r}]_{P[\hbar]},
\end{equation}
with $\alpha_{I} \in \Omega^{*}(M)$.
\end{definition}

\noindent As this classical limit does not effect the symmetries or the Jacobiators, the Schouten brackets do form a genuine $L_{\infty}$-algebra. The multi-derivation property becomes

\begin{eqnarray}
\nonumber (\alpha_{1}, \cdots, \alpha_{r}\alpha_{r+1})_{P} &=& (\alpha_{1}, \cdots \alpha_{r})_{P}\alpha_{r+1}  \pm \alpha_{r} (\alpha_{1}, \cdots , \alpha_{r+1})_{P}\\
 &+& \lim_{\hbar\rightarrow 0} \hbar (\alpha_{1}, \cdots, \alpha_{r}, \alpha_{r+1})_{P},
 \end{eqnarray}

 \noindent which clearly gives the strict Leibnitz rule.\\

\begin{theorem} The Schouten brackets can be cast into the form
\begin{eqnarray}\label{higher schouten}
\nonumber(\alpha_{1}, \alpha_{2},\cdots, \alpha_{r} )_{P} &=& [\alpha_{1}, \alpha_{2}, \cdots , \alpha_{r}]_{\stackrel{r}{P}}\\
&=& \left.\{\cdots \{ \{K_{P}, \alpha_{1}  \}, \alpha_{2} \},\cdots  ,\alpha_{r} \}\right|_{\Pi TM},
\end{eqnarray}
where $K_{P}\in C^{\infty}(T^{*}(\Pi TM))$ is a higher Schouten structure, i.e. it is odd and $\{K_{P}, K_{P} \}=0$. Here $\{ \bullet, \bullet \}$ is the canonical Poisson bracket on $T^{*}(\Pi TM)$.
\end{theorem}

\newpage

\begin{proof}
Let us work with natural local coordinates $\{ x^{A}, dx^{A}, p_{A}, \pi_{A} \}$ on $T^{*}(\Pi TM)$. We define the total symbol of the Lie derivative along the higher Poisson structure as
\begin{equation}
\sigma L_{P} = K_{P} = L_{P}(x,dx, p , \pi) \in C^{\infty}(T^{*}(\Pi TM)),
\end{equation}
via $\frac{\partial}{\partial x^{A}} \leftrightarrow p_{A}$ and $\frac{\partial}{\partial dx^{A}} \leftrightarrow \pi_{A}$. As the Lie derivative transforms as a tensor under morphisms induced by $\Diff(M)$ the function $K_{P}$ is well defined between natural local coordinates.\\

\noindent From the general theory of differential operators and microlocal analysis, see for example H$\ddot{\textnormal{o}}$rmander \cite{Hormander:1985III} we know that the symbol map takes commutators of operators to Poisson brackets of functions. In the current situation this extends to the total symbol due to the tensor nature of the Lie derivative between natural local coordinates. Then we have
\begin{equation}
\sigma [L_{P}, L_{P}] = \{ K_{P}, K_{P} \} =0,
\end{equation}
thus $K_{P} \in C^{\infty}(T^{*}(\Pi TM))$ is a higher Schouten structure.\\

\noindent Then consider
\begin{equation}
\nonumber \lim_{\hbar \rightarrow 0} \left[\cdots\left[ \left[L_{P[\hbar]}, \alpha_{1}  \right], \alpha_{2} \right], \cdots, \alpha_{r} \right]\mathds{1} = \left[\cdots\left[ \left[L_{\stackrel{r}{P}}, \alpha_{1}  \right], \alpha_{2} \right], \cdots, \alpha_{r} \right]\mathds{1}.
\end{equation}
\noindent As the above is a differential  form, the action on the unit form is via multiplication and as such it can be dropped. Then taking the total symbol gives
\begin{equation}
\left[\cdots\left[ \left[L_{\stackrel{r}{P}}, \alpha_{1}  \right], \alpha_{2} \right], \cdots, \alpha_{r} \right] = \left\{\cdots\left\{ \left\{\sigma L_{\stackrel{r}{P}}, \alpha_{1}  \right\}, \alpha_{2} \right\}, \cdots, \alpha_{r} \right\},
\end{equation}
\noindent Putting this together we see that
\begin{equation}
\lim_{\hbar \rightarrow 0  } [\alpha_{1}, \alpha_{2},\cdots , \alpha_{r}]_{P[\hbar]} = \left. \left\{\cdots\left\{ \left\{ K_{P}, \alpha_{1}  \right\}, \alpha_{2} \right\}, \cdots, \alpha_{r} \right\}\right|_{\Pi TM \subset \: T^{*}(\Pi TM)},
\end{equation}
and the result is established.
\proofend
\end{proof}

\noindent Via these constructions it is straight forward to see that

\begin{equation}
K_{P} = \sum_{r=0}^{\infty} \left( \frac{1}{r!} dx^{B}\frac{\partial P^{A_{1} \cdots A_{r}}}{ \partial x^{B}} \pi_{A_{1}} \cdots \pi_{A_{r}} - \frac{1}{(r-1)!}P^{A_{1} \cdots A_{r}} \pi_{A_{r}} \cdots \pi_{A_{2}} p_{A_{1}}\right),
\end{equation}

\noindent This function can then be pulled-back to $T^{*}(\Pi T^{*}M)$, via the canonical double vector bundle morphism $R :T^{*}(\Pi T^{*}M) \rightarrow T^{*}(\Pi TM) $. Via natural local coordinates $\{ x^{A}, x^{*}_{A}, p_{A}, \pi^{A}\}$ on $T^{*}(\Pi T^{*}M)$, the canonical morphism is given by $R^{*}(dx^{A})= (-1)^{\widetilde{A}}\pi^{A}$ and $R^{*}(\pi_{A})= x^{*}_{A}$. Thus we have

\begin{equation}
R^{*}K_{P} = \left( \sum_{r=0}^{\infty}   \frac{(-1)^{\widetilde{B}}}{r! }   \frac{\partial P^{A_{1} \cdots A_{r}}}{\partial x^{B}} x^{*}_{A_{r}} \cdots x^{*}_{A_{1}}  \right)\pi^{B} - \left( \sum_{r=0}^{\infty} \frac{1}{(r-1)!}P^{B A_{2} \cdots A_{r}} x^{*}_{A_{r}} \cdots x^{*}_{A_{2}} \right)p_{B},
\end{equation}

\noindent which is immediately recognised as the \emph{linear Hamiltonian}  associated with the homological vector field

\begin{equation}
Q_{P} = - \SN{P,\bullet} \in \Vect(\Pi T^{*}M).
\end{equation}

\noindent Recall that an $L_{\infty}$-algebroid structure on a vector bundle $E \rightarrow M$ is a homological vector field on $\Pi E$. Thus, there is an  $L_{\infty}$-algebroid structure on $T^{*}M$ provided by the above homological vector field (this first appears in \cite{khudaverdian-2008}).  This is the higher analogue of the  Lie algebroid associated with a Poisson manifold.\\

\noindent In essence, Khudaverdian \& Voronov reverse this process and define a map $P \rightsquigarrow K_{P}$ and use this in the definition of the Schouten brackets viz Eqn.(\ref{higher schouten}). We see that the Schouten brackets of Khudaverdian \& Voronov are in  correspondence with the higher Koszul--Schouten brackets.

\begin{remark}
Theorem (\ref{relation poisson koszul--schouten}) carries over directly to the Schouten brackets. This can be understood by counting form degrees and realising that the Koszul--Schouten and Schouten brackets coincide when restricted to functions and exterior derivatives of functions as required by the theorem. Then (up to conventions) this theorem is identical to a theorem found in  \cite{khudaverdian-2008}.
\end{remark}

\section{Concluding Remarks}

\noindent The construction of higher Poisson and higher Koszul--Schouten brackets represent a geometric generalisation of the classical brackets found in Poisson geometry and classical mechanics. That is we consider the natural generalisation of a higher order multivetor field that Schouten--Nijenhuis self-commutes and derive the theory from this starting place.   Laying behind this is the algebraic theory of higher derived brackets of Voronov \cite{voronov-2004,voronov-2005} which leads us to the theory of $L_{\infty}$-algebras viz Lada \& Stasheff \cite{Lada:1992wc}. \\

\noindent  The classical master equation is   a  ``differential geometric condition". Thus, from a geometric point of view higher Poisson structures are quite natural, as compared to say  Nambu--Poisson structures \cite{Nambu:1973qe,Takhtajan:1993vr}. Despite this, the physical applications of higher Poisson structures  as a generalised setting for classical mechanics is unclear. This is before one even begins to consider quantisation, which if possible would require some deviation from established methods.   \\

\noindent In this note we showed that the Koszul--Schouten bracket found in classical Poisson geometry generalises to the higher case, including the morphism between the Poisson and Koszul--Schouten brackets provided by the exterior derivative. We showed how they relate to the Schouten brackets found in \cite{khudaverdian-2008}.\\

\noindent What other aspects of classical Poisson geometry carry over to the higher case awaits to be explored. It should be possible to consider a generalised version of Poisson homology and cohomology for higher Poisson manifolds, for example. For the homogenous case initial work was presented by de Azc$\acute{\textnormal{a}}$rraga et.al \cite{deAzcarraga:1996jk}.\\

\noindent The structure of a higher Poisson manifold can be found lying behind the classical Batalin–-Vilkovisky antifield formulism \cite{Batalin:1981jr,Batalin:1984jr}, modulo extra gradings. Indeed, the initial motivation for this work lies in wanting a deeper understanding of the BV-formulism and the relation with classical differential geometry.\\

\noindent Thus, we see that the space of field-antifield valued differential forms comes equipped with the structures of a homotopy BV-algebra, or equivalently a homotopy Schouten algebra. However, the role that such higher brackets play in field theory awaits to be explored.

\section*{Acknowledgments}

\noindent It is a pleasure to thank the organisers of the first HEP Young Theorists' Forum at  University  College London where this work was first reported.   The author would like to thank Drs Th.Th. Voronov and H.M. Khudaverdian for many interesting discussions and guidance. A special thank you also goes to Dr Klaus Bering  for his interest in earlier versions of this work. This work was funded by  the British tax payer via an EPSRC DTA.

\begin{small}
\bibliography{preprintbib}
\end{small}
\vfill

\begin{center}
\noindent SCHOOL OF MATHEMATICS, THE UNIVERSITY OF MANCHESTER, UNITED KINGDOM, M13 9PL.\\
\noindent \emph{e-mail address}: \texttt{andrewjames.bruce@physics.org}
\end{center}
\end{document}